\documentclass[pra,twocolumn,aps,superscriptaddress]{revtex4-1}
\usepackage{amsmath,amsfonts,amssymb,amsthm,graphicx,bbm,enumerate,times,color}
\usepackage{mathtools,comment}
\usepackage{lipsum}
\newtheorem{theorem}{Theorem}

\newtheorem{lemma}[theorem]{Lemma}

\newtheorem{observation}[theorem]{Observation}

\usepackage{hyperref}

\newcommand{\e}{{\rm e}}

\newcommand{\id}{\mathbbm{1}}

\newcommand{\ket}[1]{|#1\rangle}

\newcommand{\Tr}{\operatorname{tr}}
\newcommand{\tr}{\Tr}

\newcommand{\fu}{Dahlem Center for Complex Quantum Systems, Freie Universit{\"a}t Berlin, 14195 Berlin, Germany}
\newcommand{\PIK}{Potsdam Institute for Climate Impact Research, 14473 Potsdam, Germany}
\newcommand{\hu}{Department of Physics, Humboldt-Universit{\"a}t zu Berlin, 12489 Berlin, Germany}

\begin{document}
\title{Quantum thermodynamics with local control}

  \author{J. Lekscha}
    \affiliation{\fu}
      \affiliation{\PIK}
        \affiliation{\hu}

  \author{H. Wilming}
  \affiliation{\fu}
  \author{J. Eisert}
  \affiliation{\fu}
  \author{R. Gallego}
  \affiliation{\fu}

\begin{abstract}
We investigate the limitations that emerge in thermodynamic tasks as a result of having local control only over the components of a thermal machine. These limitations are particularly relevant for devices composed of interacting many-body systems. Specifically, we study protocols of work extraction that employ a many-body system as a working medium whose evolution can be driven by tuning the on-site Hamiltonian terms. This provides a restricted set of thermodynamic operations, giving rise to novel bounds for the performance of engines. Our findings show that those limitations in control render it in general impossible to reach Carnot efficiency; in its extreme ramification it can even forbid to reach a finite efficiency or finite work per particle. We focus on the 1D Ising model in the thermodynamic limit as a case study. We show that in the limit of strong interactions the ferromagnetic case becomes useless for work extraction, while the anti-ferromagnetic improves its performance with the strength of the couplings, reaching Carnot in the limit of arbitrary strong interactions. Our results provide a promising connection between the study of quantum control and thermodynamics and introduce a more realistic set of physical operations well suited to capture current experimental scenarios.  
\end{abstract}
\date{\today}
\maketitle

Recently, notions of  quantum thermodynamics, and in particular questions on how much 
work can be extracted in systems in which quantum effects are expected to be relevant, 
has received a lot of attention. Much focus has been put on 
understanding fundamental limits on the amount of work that can be extracted
from a single quantum system prepared in a state
out of thermal equilibrium. 
This research programme is two-pronged: On the one hand, there is an emphasis on
identifying the laws of quantum thermodynamics \cite{Skrzypczyk2010,Negative,Abergtrullywork,Machines,
Secondlaw,ResourceTheory,Negative}, as primitives from which macroscopic thermodynamics can be derived.
On the other hand, a significant body of literature is concerned with 
characterising the behaviour of realistic physical devices 
operating at a scale where quantum effects become relevant \cite{Alicki79,Kosloff84,Gemmer,ControlGelbwaser,Machines}. 

One of the key aspects in these efforts 
is to understand how quantum thermodynamic notions precisely 
behave under composition of subsystems. This comprises the study of the role of correlations between subsystems \cite{Oppenheim02,AlickiFannes13,Acinentanglement,Perarnau15}, 
of the scalability of quantum engines \cite{Campisi16,Campisi16b,Zheng15}, and of 
the emergence of thermodynamics from a more fundamental quantum description of its constituents \cite{Janzing00,Gemmer}. This body of literature focuses on composite systems that, although often displaying classical or even quantum correlations between its subsystems reflecting a past interaction, do 
not interact, or at least not beyond the weak-coupling regime. 

In this work, we contribute to filling 
this gap by focusing on the study of work extraction with many-body systems with possibly strong couplings between subsystems. 
An important question that emerges when dealing with such strongly interacting systems is that of determining the possible transformations that one can induce in the state of the compound by having \emph{local control} only \footnote{If the subsystems do not interact -- or they do it in the weak-coupling regime -- the set of possible dynamics is trivially given by products of local unitaries. In this regime the limitations of local control are not present.}. A reasonable setting
for a many-body system is one where the experimenter will be able to apply and vary external fields that will control the on-site Hamiltonian terms; at the same time, the interaction terms between the subsystems cannot be modified at will. 
This constitutes a limitation on the set of reachable Hamiltonians and consequently on the possible dynamics that the system may undergo. This is a most natural setting: The field of \emph{quantum control} (QC) can be seen
as largely studying the type of dynamics precisely in such a setting \cite{Lloyd1995,Ramakrishna1995,Albertini2002,Janzing2002,Lloyd2004,Burgarth2009}. Here, 
we explore the surprising ramifications of local control for the performance of thermodynamic tasks. We believe that the
identification of this physically reasonable class of thermodynamic state transformations constitutes an
important aspect of this work in its own right.

We introduce interactions and limitations on control into the problem of work extraction by considering the situation of an engine operating with a many-body system as a working medium and two thermal baths at different temperatures. The working medium has some fixed interactions of arbitrary strength among its constituents. The engine is operated by applying some time-dependent external fields and putting the working medium in contact with either of the baths. For 
this general scenario, we investigate the limitations emerging due to the lack of global control, simply by comparing with the usual bounds provided by the second law. Those limitations will affect the efficiency of the engine as a function of the interactions and the size of the many-body system. 

As a first result, we find a fully general expression describing the corrections to the Carnot efficiency as a function of the interactions, showing that it is impossible in general to achieve Carnot efficiency exactly and interactions lead to 
irreversibility in the thermodynamic sense. Also, by employing results from the theory of QC, 
we show that our bounds are saturated for generic interactions. We then elaborate on bounds for the 1D Ising model as a case study. Surprisingly, this model displays a strikingly different behaviour for the anti-ferromagnetic and the ferromagnetic regimes. The former case allows for a finite work output per particle, as well as reaching Carnot efficiency in the limit of very strong couplings. The latter displays an opposite behaviour, where very strong couplings imply vanishing work per particle and efficiency. This shows that limitations due to local control crucially affect the scalability or performance in the macroscopic limit. Indeed, ranging from the two extreme behaviours of allowing for Carnot efficiency or preventing one to extract any work whatsoever.

\emph{Set-up and operations considered.} 
We consider a thermal machine composed by a working medium and two baths at different temperatures. 
The working medium is taken to be a many body system composed of $N$ subsystems. 
The machine is operated by performing two kinds of operations. 

Firstly, one can \emph{change the Hamiltonian} of the working medium over time. The working medium is hence described at time $t$ by the pair $(\rho(t),H(t))$ of a quantum state and a time-dependent Hamiltonian
\begin{equation}\label{eq:generalhamiltonian}
H(t)=H_{\text{ext}}(t)+ H_{\text{int}}.
\end{equation}
The term $H_{\text{ext}}(t)$ represents the external fields 
\begin{equation}
H_{\text{ext}}(t) =\sum_j^N H^{(j)}(t)
\end{equation}
that can be \emph{varied with time},
where $H^{(j)}=\id_1 \otimes \cdots \otimes h^{(j)}(t) \otimes \cdots \otimes \id_N $ is a Hamiltonian acting on the $j$-th subsystem only. Clearly, it is of physical relevance to consider the special case in which $h^{(j)}(t)=h(t)$ $\forall j$, that is, where the external fields act equally in all subsystems. This will be indeed the case considered in our case studies, but we keep here the discussion as general as possible. 
$H_{\text{int}}$, in contrast, is an arbitrary \emph{time-independent interaction} between the subsystems. The form of the interactions between the constituents of the working medium will be crucial in this step, because they shape the limitations on the set of Hamiltonians that can be chosen. This limitation is the most natural when dealing with many-body systems that can be affected by controlled external fields, although their interactions are not accessible to the experimenter. 
The time evolution between times $t_1$ and $t_2$ of the working medium under Hamiltonian \eqref{eq:generalhamiltonian} results in a transition 
\begin{equation}\label{eq:timeevolution}
\big(\rho(t_1),H(t_1)\big) \mapsto \big(U(t_1,t_2)\rho(t_1)U^{\dagger}(t_1,t_2), H(t_2)\big)  
\end{equation}
where $U(t_1,t_2)$ is the unitary evolution induced by the time-dependent Hamiltonian \eqref{eq:generalhamiltonian}. This transition results in a change of the total mean energy of the working medium, which is the 
\emph{expected work} $W$ extracted in the process, so that
\begin{equation}
W(t_1,t_2)=\tr \left(\rho(t_1) H (t_1)\right) - \tr \left(\rho(t_2) H(t_2)\right),
\end{equation}
where $\rho(t_2)=U(t_1,t_2)\rho(t_1)U^{\dagger}(t_1,t_2)$.  Given that the current value of the time employed will not be relevant for work and efficiency considerations, we can just describe the processes by 
\begin{eqnarray}
W^{i,i+1}&:=& W(t_i,t_{i+1}),\,\,\,
(\rho^i,H^i):= (\rho(t_i),H(t_i)).
\end{eqnarray} 

Secondly, we will consider another kind of operations that represent the \emph{thermal contact} between the working system and a thermal bath at inverse temperature $\beta$. We will assume that throughout the protocol there are two different baths available, one hot bath and one cold bath with inverse temperatures given by $\beta_h$ and $\beta_c$ respectively. These operations have the effect of bringing the working medium to the Gibbs state of the corresponding Hamiltonian. That is,
\begin{equation}\label{eq:thermalcontact}
(\rho^i,H^i) \mapsto (\omega(H^i,\beta),H^i),
\end{equation}
where $\omega(H^i,\beta)=\exp(-\beta H^i)/Z_{\beta}(H^i)$ and the partition function is given by $Z_{\beta}(H^i)=\tr(\exp(-\beta H^i))$. To simplify the notation we will simply denote Gibbs states by $\omega_{h}^{i}:=\omega(H^i,\beta_h)$ and equivalently for the cold bath. The transformation given by \eqref{eq:thermalcontact} occurs when placing a sufficiently weak interaction between the working medium and the thermal bath and does not require any work investment. 
These two kinds of steps are repeated at will to perform a protocol. 


\emph{Engine cycles.} The two kinds of operations described above -- time evolution under the time-dependent Hamiltonian \eqref{eq:generalhamiltonian} and thermal contacts with the two baths resulting in \eqref{eq:thermalcontact} -- are combined arbitrarily in a protocol yielding a total expected work given by 
\begin{equation}
W=\sum_{i}^{n}W^{i,i+1},
\end{equation}
where $i$ sums over all the steps in which a time evolution under the time-dependent Hamiltonian has occurred. The protocol is applied cyclicly so that after the $n$ steps that change the Hamiltonian we return to the initial Hamiltonian, that is, $H^n=H^0$. 
During the contacts with the thermal baths resulting in \eqref{eq:thermalcontact} the bath and working medium exchange heat. The heat provided by the hot thermal bath is given by
\begin{equation}\label{eq:heatdef}
Q_h=\sum_{i}^k Q^{i}_h=\sum_{i}^k  \left(\tr(\omega_{h}^i H^i) - \tr(\rho^i H^i)\right),
\end{equation}
where $i$ sums over the steps of the protocol where a thermal contact with the hot thermal bath is implemented. 
Finally, the efficiency of the engine performing a given cycle is defined as
\begin{equation}
\eta= \frac{W}{|Q_h|}.
\end{equation}
We will now study limitations on the maximal efficiency achievable given as a function of the interaction term $H_{\text{int}}$.

\emph{Limitations to Carnot efficiency.}  As it is clear from basic considerations in (quantum) thermodynamics, 
the optimal efficiency is reached by reversible protocols. This can be easily appreciated in the Carnot cycle as depicted in Fig \ref{fig:carnotcycle}. Within the framework of phenomenological thermodynamics the working medium (say, a 
gas in a piston) is described by its entropy and temperature. It is necessary, in order to perform a Carnot cycle, that an adiabatic compression/expansion of the gas in the piston can alter its temperature at will within the range given by the two baths. That is, if one has a gas at temperature $T_c$ (after contact with the cold bath) one can compress rapidly the piston to increase its temperature to $T_h$, the latter being the temperature of the hot bath. The temperature will increase monotonically with the strength of the compression. Hence, in order to reach $T_h$ one only needs to compress the gas sufficiently.  
\begin{figure}
\includegraphics[width=.82\linewidth]{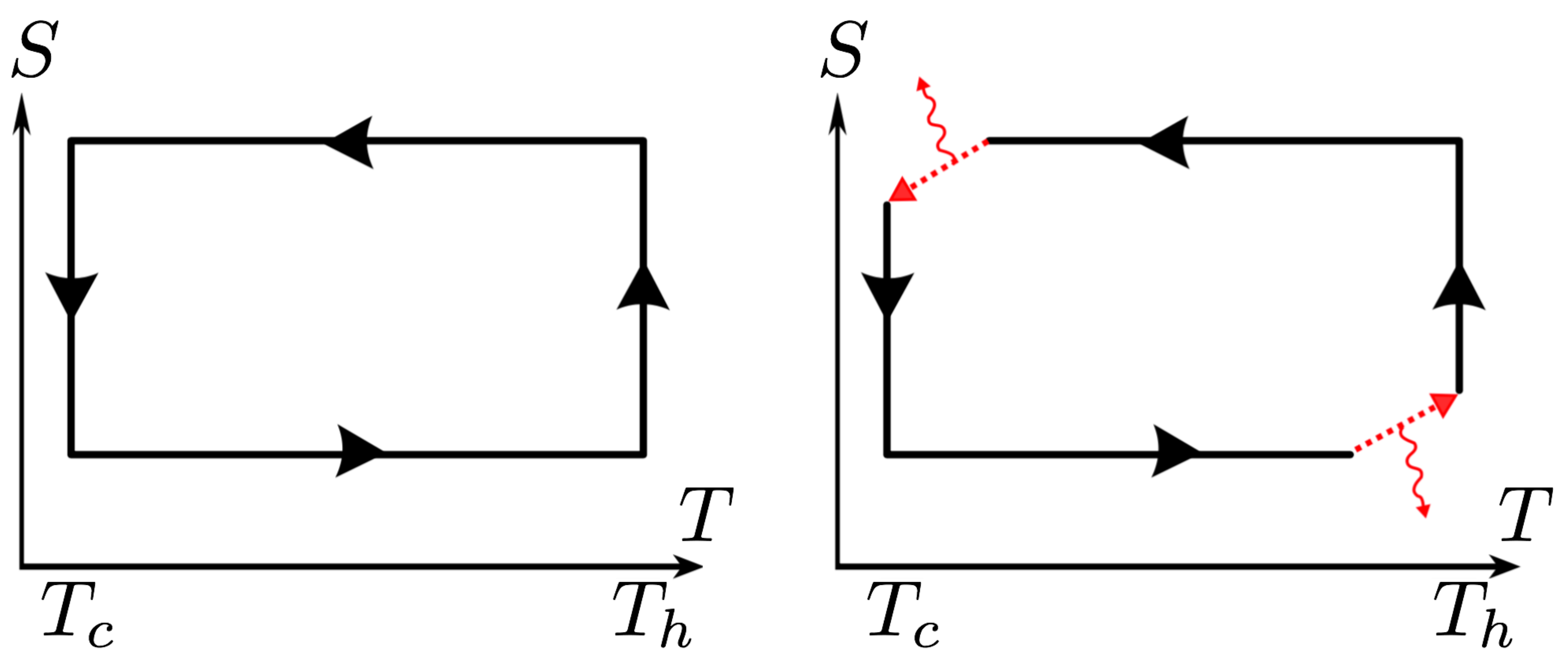}
\caption{\textbf{a}) A Carnot cycle as formulated within the framework of phenomenological thermodynamics in an entropy-temperature diagram. During the adiabatic compression/expansion (horizontal) the temperature of the working medium changes in such a way that when it is put in contact with the heat bath it is already at the same temperature as the bath. Hence, the thermal contact that initiates the isothermal has no effect on the working medium. No heat flows from the heat bath until one starts the isothermal expansion. \textbf{b}) Suppose that by some technical limitation the temperature of the bath cannot be reached by an adiabatic expansion/compression. Then when the working medium and the baths are put in contact, there is an unavoidable dissipation reducing the efficiency, illustrated by the red arrows.}
\label{fig:carnotcycle}
\end{figure}

The idealisation of a Carnot engine is similar when we deal with a microscopic working medium. In this case, it will not be described by the coarse-grained variables entropy and temperature, but with the pair $(\rho^i,H^i)$ of the 
quantum state $\rho^i$ and the Hamiltonian $H^i$, taking different configurations over the protocol. The diagram of this state space is depicted in Fig.\ \ref{fig:carnotcyclequbit}. 

\begin{figure}
\includegraphics[width=.98\linewidth]{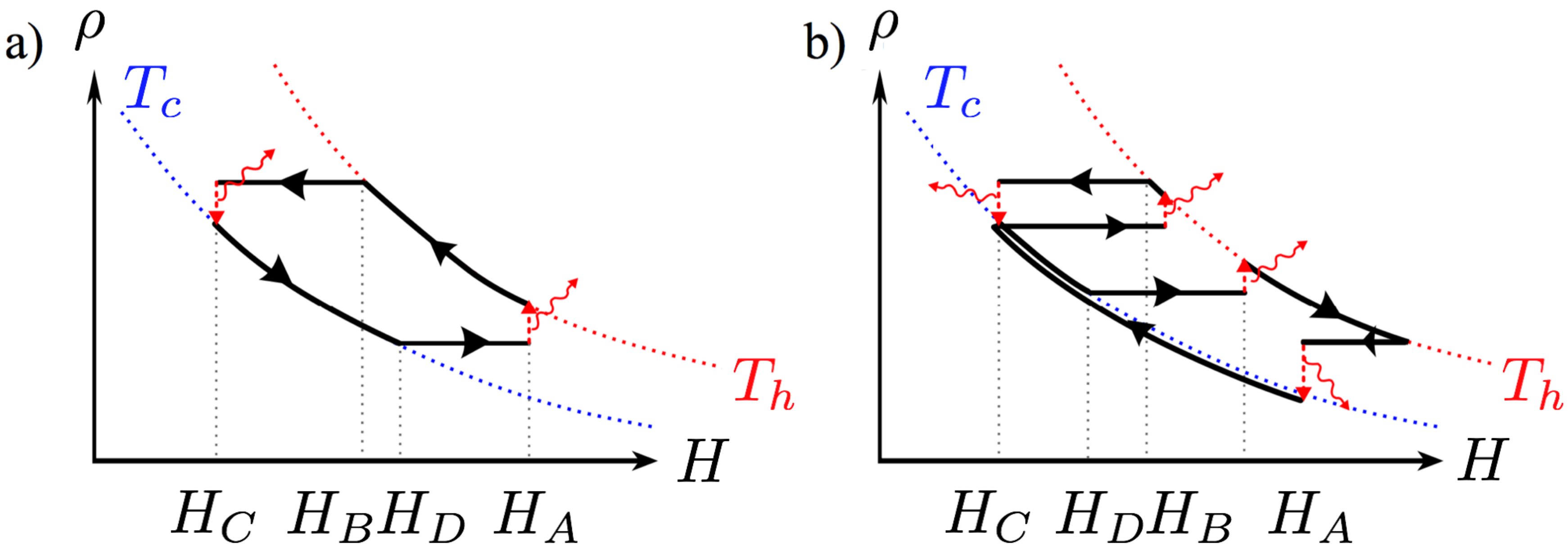}
\caption{{\bf a)} A Carnot-like protocol for a quantum system. {\bf b)} A protocol that we do not call a Carnot-like protocol. 
The red (blue) dashed lines depict the thermal states at the hot (cold) bath temperature.}
\label{fig:carnotcyclequbit}
\end{figure}

Similarly to the usual Carnot cycle of Fig.\ \ref{fig:carnotcycle}, maximal efficiency is achieved when the protocol is reversible. This requires that the working medium, after contact with the cold heat bath in state $\omega^i_c(H)$ can be transformed by an adiabatic process -- in this case one of the form Eq.\ \eqref{eq:timeevolution} -- into a state $\omega^i_h(H')$. In other words, one must be able to ``compress'' the cold working medium until it reaches the temperature of the hot bath. This is always possible for a gas in a piston, but as we will show \emph{it is not possible} for many-body systems evolving under Hamiltonians of the form \eqref{eq:generalhamiltonian}. This insight constitutes the main result of this manuscript and it is responsible for the impossibility of reaching Carnot efficiency, which is captured precisely in the following set of results.


\emph{General bound.} We will first consider the problem in its full generality and give an upper bound to the efficiency that can be obtained by a protocol that combines operations \eqref{eq:timeevolution} and \eqref{eq:thermalcontact}.

\begin{theorem}[General bound]\label{thm:general-results}
All protocols that combine operations \eqref{eq:timeevolution} and \eqref{eq:thermalcontact} have an efficiency bounded as
\begin{align}\label{eq:mainbound}
\eta &\leq 1 - \dfrac{T_c}{T_h} \left(\dfrac{\Delta S^{B,D} + \min_{U} D(U\omega_{h}^{B}U^{\dagger}\|\omega_{c}^{C}) }{\Delta S^{B,D} -  \min_{V}D(V\omega_{c}^{D}V^{\dagger}\|\omega_{h}^{A})}\right)
\end{align}
where $\omega_{i}^J:=\omega(H^J,\beta_i)$ and $H^J$ are arbitrary Hamiltonians of the form \eqref{eq:generalhamiltonian}; $\Delta S^{B,D}:=S(\omega_h^B)-S(\omega_c^D)$ where $S$ is the von Neumann entropy, $D(\cdot \| \cdot)$ is the relative entropy; $U$ and $V$ are unitary transformations that can be induced by any trajectory of $H(t)$ as in \eqref{eq:generalhamiltonian}.
\end{theorem}

The proof is presented in detail in Appendix \ref{sec:app:penalty}. It  also shows how to construct for any $U,V$ and Hamiltonians $H^{A,B,C,D}$ a protocol that actually saturates the bound. This protocol is what we call a Carnot-like protocol, which contains one isothermal path with each of the baths and two adiabatic operations of the form \eqref{eq:timeevolution}, but differs from a conventional Carnot protocol in the fact that there is an unavoidable dissipation when initiating the isothermals (see Fig.~\ref{fig:carnotcyclequbit})

It is also worth discussing simplified bounds on the efficiency that can be also saturated in two relevant regimes. First, note that if $[H_\text{ext}(t),H_\text{int}]=0$ $\forall t$, then any unitary $U$ generated by a trajectory of $H(t)$ will be such $[U,\omega^J_i]=0$. In this case \eqref{eq:mainbound} is replaced by
\begin{align}\label{eq:mainbounddiagonal}
\eta &\leq 1 - \dfrac{T_c}{T_h} \left(\dfrac{\Delta S^{B,D} +  D(\omega_{h}^{B}\|\omega_{c}^{C}) }{\Delta S^{B,D} -  D(\omega_{c}^{D}\|\omega_{h}^{A})}\right).
\end{align} 
This latter bound applies also if one considers arbitrary Hamiltonians $H(t)$ but limits instead the set of operations \eqref{eq:timeevolution} to Hamiltonian quenches, where $U(t_1,t_2)=\id$.

Secondly, consider the case where a trajectory of $H(t)$ can induce any possible global unitary transformations $U$ and $V$. In this case the bound \eqref{eq:mainbound} is replaced by 
\begin{align}\label{eq:mainboundany}
\eta &\leq 1 - \dfrac{T_c}{T_h} \left(\dfrac{\Delta S^{B,D} +  D^\downarrow(\omega_{h}^{B}\|\omega_{c}^{C}) }{\Delta S^{B,D} -  D^\downarrow(\omega_{c}^{D}\|\omega_{h}^{A})}\right)
\end{align} 
where $D^\downarrow(\cdot \| \cdot )$ is the relative entropy defined as
$D^\downarrow(\rho||\sigma):=\sum_{m}\rho_m \ln({\rho_m}/{\sigma_m})$,
with $\{\rho_m\}$ and $\{\sigma_m\}$ being the set of eigenvalues of $\rho$ and $\sigma$ respectively, both ordered in non-increasing order. The bound \eqref{eq:mainboundany} follows from majorization arguments (see Appendix \ref{sec:app:relent}) and its a universal bound on the efficiency (it is larger than the r.h.s. of \eqref{eq:mainbound}). We will also show further in this manuscript that this universal bound can be achieved when having generic interactions.
Indeed, Theorem \ref{thm:general-results} allows us to recover the usual Carnot efficiency in the case of vanishing interactions, since in this case the correction terms can be made zero by appropriate choice of local fields.
\begin{observation}[Vanishing interactions]
In the case of vanishing interactions, that is $H_{\text{int}}=0$ in \eqref{eq:generalhamiltonian} it is possible to achieve Carnot efficiency 
\begin{equation}\label{eq:carnot efficiency}
\eta_c=1- \frac{T_c}{T_h}.
\end{equation}
\end{observation}
This follows simply from Theorem \ref{thm:general-results}, since 
\begin{equation}
\omega_h^B:=\bigotimes_i \omega(H^B_i,\beta_h)=\bigotimes_i \omega(H^C_i,\beta_c):=\omega_c^C
\end{equation} 
can be satisfied by taking simply $H_i^C=({\beta_h}/{\beta_c}) H_i^B$, and equivalently for $\omega_h^D$ and $\omega_h^A$. By choosing $U=V=\id$ in \eqref{eq:mainbound} we obtain the Carnot efficiency \eqref{eq:carnot efficiency}. 
%
Finally, we note that the correction terms in the optimal efficiency scale extensively for local many-body systems. Hence, a similar bound holds when we consider for the efficiency the work-density and heat-density instead of the total work and total heat.

\emph{Saturating the bound for generic interactions.} Previously we have seen that the maximal possible value of the efficiency as a function of the interactions $H_\text{int}$ is given by \eqref{eq:mainboundany}. It is a natural task 
to establish conditions where it can be saturated. 
Here we argue that the \emph{bound \eqref{eq:mainboundany} can be saturated generically.}
For this, we rely on results in the field of \emph{quantum control} showing that under a Hamiltonian of the form \eqref{eq:generalhamiltonian}, any unitary in the Lie-algebra generated 
by $H_{\rm int}$ and the locally controllable fields $H^{(i)}(t)$ can be approximated arbitrarily well \cite{Lloyd1995,Ramakrishna1995,Albertini2002,Janzing2002,Lloyd2004,Burgarth2009}. For generic, locally interacting Hamiltonians, this Lie-algebra is the full special unitary Lie-algebra on the Hilbert-space and thus any global unitary can in principle be approximated arbitrarily well.  Indeed, often one does not even need to control the on-site field of all the spins. For example, in a spin-chain with Heisenberg-like interactions, control over a \emph{single} spin is in principle sufficient to implement any unitary evolution \cite{Janzing2002,Burgarth2009}.

\emph{Case study: Ising model.} For the remainder of this paper we will focus on the Ising model as a case study, it
being instructive and sharing all the main features discussed here. The goal is to study the limitations to Carnot derived in Theorem \ref{thm:general-results} from a quantitative perspective. We will show that the corrections to Carnot influence dramatically the feasibility of work extraction protocols and that one encounters a remarkably rich variety of behaviours. This ranges from situations where strong interactions make impossible to extract any work per particle at all (ferromagnetic) to the case of strong interactions enhancing the efficiency to Carnot (anti-ferromagnetic). We study work extraction from a many-body spin system with nearest neighbour Ising Hamiltonian, for which Eq.\ \eqref{eq:generalhamiltonian} 
 takes the form
\begin{align}
H_{{\rm I}, N}(t) &= -h(t) \sum_{j=1}^N \sigma_z^{(j)} - J \sum_{j=1}^N \sigma_z^{(j)} \sigma_z^{(j+1)}. \label{equ:def-H-Ising}
\end{align}
Here, $\sigma_z^{(j)}$ denotes the Pauli-Z-matrix acting at spin $j$ and $h(t)$ is a tunable magnetic field. Note, that we are assuming that the external field $h(t)$ is translational invariant and commutes with the interaction. Therefore, the unitaries $U,V$ in~\eqref{eq:mainbounddiagonal} are in fact identities. We thus obtain a set of operations less general as the one given by \eqref{eq:generalhamiltonian}, but at the same time it fairly represents a more realistic situation than applying different external fields to each microscopic subsystems. The
interaction strength $J$ is fixed and models the experimentally not controllable interaction between two neighbouring spins. 
We assume periodic boundary conditions, i.e., $\sigma_z^{(N+1)} = \sigma_z^{(1)}$. As the Hamiltonian~\eqref{equ:def-H-Ising} is
diagonal, it is equivalent to the classical Ising model Hamiltonian with 
$\sigma_z^{(j)} = \sigma^{(j)} \in \left\lbrace -1, 1\right\rbrace$ denoting spin up or down respectively.  
Thus, we will be able use the well known results about the partition function of the Ising model 
when studying work extraction. 
Using the partition function and the bounds of Theorem \ref{thm:general-results} one can compute the efficiency at maximum work density as function of $J$. This is shown in Fig.\ \ref{fig:maxefficiency} in the thermodynamic limit.

\begin{figure}
\includegraphics[width=.7\linewidth]{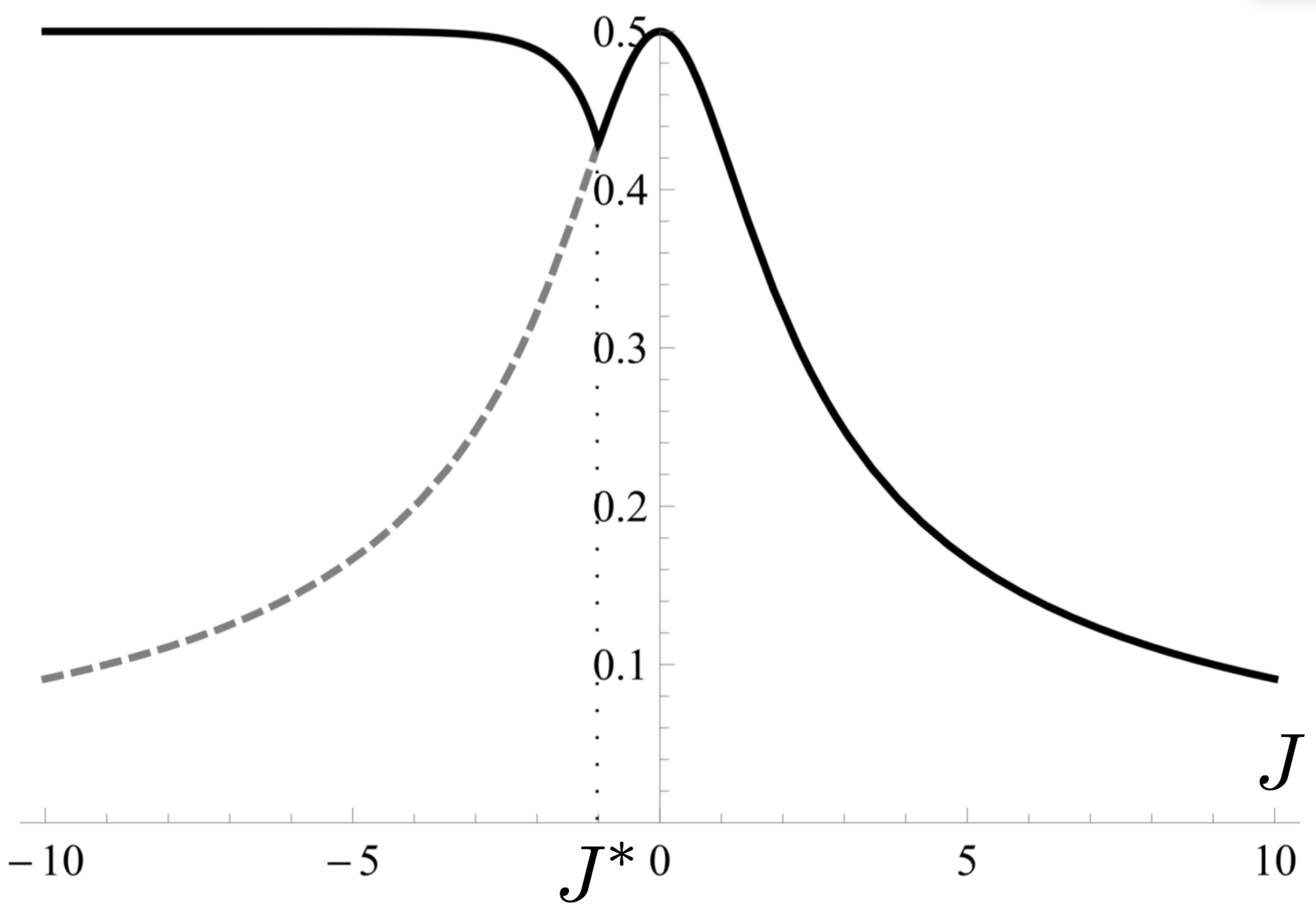}
\caption{Efficiency at maximum work density (black) for the Ising model as a function of $J$ in the thermodynamic limit for the parameters $\beta_h=0.5,\beta_c=1$. The gray dashed line shows a protocol that is independent of $J$.}
\label{fig:maxefficiency}
\end{figure}

There are three relevant aspects of the efficiency plotted in Fig.\ \ref{fig:maxefficiency} that we can 
derive analytically and that clearly exemplify the behaviour of the bound \eqref{eq:mainbound}: $i$) in the limit of strong anti-ferromagnetic interactions $J\rightarrow -\infty$ one can reach Carnot efficiency, $ii$) in the limit of strong ferromagnetic interactions $J\rightarrow \infty$ efficiency drops to zero and $iii$) at the value $J=J^*$, the efficiency changes its behaviour abruptly. All these three points can be explained analytically relying on considerations about the ground state degeneracy of the (anti)-ferromagnetic regimes and the resulting optimal protocols (see Appendix \ref{sec:ground state-deg}).  

Regarding point $i$) we construct in Appendix \ref{sec:anti-ferromagnetic} a simple protocol that reaches Carnot in the limit of $J\rightarrow -\infty$. Furthermore, this protocol is able to reach a finite work output per particle and cycle. This protocol works also in the thermodynamic limit, hence showing that one can build effective engines with macroscopic strongly correlated anti-ferromagnetic spin chains. 

The limit of $J\rightarrow \infty$ described in $ii$)  displays a strikingly unstable behaviour. On the one hand, formally it is possible to construct a work extraction protocol that achieves Carnot efficiency for any number of particles $N$. On the other hand, this protocol has to be considered unphysical, since it requires that the external magnetic fields are controlled with degree of precision that scales with $N$. That is, for any finite precision on the external parameter $h(t)$, one can find a sufficiently large $N$ so that the maximal efficiency vanishes. This is discussed in detail in Appendix \ref{sec:ferromagnetic} and in fact holds for more general classes of gapped ferromagnetic Hamiltonians. Since the precision on $h(t)$ is an intensive quantity, one can conclude that the strongly correlated ferromagnetic spin chains become useless as working mediums for engines in the thermodynamic limit. Indeed, we can also show that no finite work-density can be achieved in the limit $J\rightarrow \infty$. Furthermore, this unstable behaviour does not occur for the anti-ferromagnetic case $i$).
Lastly, in order to explain $iii$) we derive that the form of the optimal protocol changes abruptly at $J=J^*$, 
explaining that the efficiency is not smooth in this point in the thermodynamic limit. 
This is derived in Appendix \ref{sec:nonsmooth}.

\emph{Conclusion.} In this work, we have introduced the study of the performance of thermodynamic
engines in the presence of limited control on the thermodynamic operations. Our results complement a 
significant body of work in quantum thermodynamics concerned with the ultimate bounds on thermodynamics,
but in the absence of natural control restrictions. We derived corrections to Carnot efficiency as a result of having local control, which we introduce by considering engines driven by local external fields. Previous results from quantum control show that our general bounds are achievable for generic interactions. This opens new venues to incorporate in a 
comprehensive fashion the extensive literature on quantum control to thermodynamics.

It is also promising to investigate other possible sets of limited thermodynamic 
operations incorporating a notion of locality. This points to the possibility of developing formally a 
\emph{resource theory of locally restricted thermal operations}. Possible extensions of our formalism include more general and realistic thermal baths, as well as local unitaries instead of local external fields. Lastly, the corrections on the efficiency can be investigated for other systems than the 1D-Ising model. In particular, it would be 
interesting to understand the effects on the efficiency for systems displaying spontaneous 
magnetization for low temperatures.

\emph{Acknowledgements.} We acknowledge funding from the BMBF (Q.com), the  EU (RAQUEL, AQuS), the DFG (EI 519/7-1, CRC 183, GA 2184/2-1), the ERC (TAQ) and the Studienstiftung des Deutschen Volkes.

\begin{thebibliography}{10}
\expandafter\ifx\csname url\endcsname\relax
  \def\url#1{\texttt{#1}}\fi
\expandafter\ifx\csname urlprefix\endcsname\relax\def\urlprefix{URL }\fi
\providecommand{\bibinfo}[2]{#2}
\providecommand{\eprint}[2][]{\url{#2}}

\bibitem{Skrzypczyk2010}
\bibinfo{author}{Linden, N.}, \bibinfo{author}{Popescu, S.} \&
  \bibinfo{author}{Skrzypczyk, P.}
\newblock \bibinfo{title}{{How small can thermal machines be? The smallest
  possible refrigerator}}.
\newblock \emph{\bibinfo{journal}{Phys. Rev. Lett.}}
  \textbf{\bibinfo{volume}{105}}, \bibinfo{pages}{130401}
  (\bibinfo{year}{2010}).

\bibitem{Negative}
\bibinfo{author}{del Rio, L.}, \bibinfo{author}{Aberg, J.},
  \bibinfo{author}{Renner, R.}, \bibinfo{author}{Dahlsten, O.} \&
  \bibinfo{author}{Vedral, V.}
\newblock \bibinfo{title}{The thermodynamic meaning of negative entropy}.
\newblock \emph{\bibinfo{journal}{Nature}} \textbf{\bibinfo{volume}{474}},
  \bibinfo{pages}{61} (\bibinfo{year}{2011}).

\bibitem{Abergtrullywork}
\bibinfo{author}{Aberg, J.}
\newblock \bibinfo{title}{Truly work-like work extraction via a single-shot
  analysis}.
\newblock \emph{\bibinfo{journal}{Nature Comm.}} \textbf{\bibinfo{volume}{4}},
  \bibinfo{pages}{1925} (\bibinfo{year}{2013}).

\bibitem{Machines}
\bibinfo{author}{Mari, A.} \& \bibinfo{author}{Eisert, J.}
\newblock \bibinfo{title}{{Cooling by heating: Very hot thermal light can
  significantly cool quantum systems}}.
\newblock \emph{\bibinfo{journal}{Phys. Rev. Lett.}}
  \textbf{\bibinfo{volume}{108}}, \bibinfo{pages}{120602}
  (\bibinfo{year}{2011}).

\bibitem{Secondlaw}
\bibinfo{author}{Brandao, F. G. S.~L.}, \bibinfo{author}{Horodecki, M.},
  \bibinfo{author}{Ng, N. H.~Y.}, \bibinfo{author}{Oppenheim, J.} \&
  \bibinfo{author}{Wehner, S.}
\newblock \bibinfo{title}{The second laws of quantum thermodynamics}.
\newblock \emph{\bibinfo{journal}{PNAS}} \textbf{\bibinfo{volume}{112}},
  \bibinfo{pages}{3275} (\bibinfo{year}{2015}).

\bibitem{ResourceTheory}
\bibinfo{author}{Brandao, F. G. S.~L.}, \bibinfo{author}{Horodecki, M.},
  \bibinfo{author}{Oppenheim, J.}, \bibinfo{author}{Renes, J.~M.} \&
  \bibinfo{author}{Spekkens, R.~W.}
\newblock \bibinfo{title}{The resource theory of quantum states out of thermal
  equilibrium}.
\newblock \emph{\bibinfo{journal}{Phys. Rev. Lett.}}
  \textbf{\bibinfo{volume}{111}}, \bibinfo{pages}{250404}
  (\bibinfo{year}{2013}).

\bibitem{Alicki79}
\bibinfo{author}{Alicki, R.}
\newblock \bibinfo{title}{The quantum open system as a model of the heat
  engine}.
\newblock \emph{\bibinfo{journal}{J. Phys. A}} \textbf{\bibinfo{volume}{12}},
  \bibinfo{pages}{103} (\bibinfo{year}{1979}).

\bibitem{Kosloff84}
\bibinfo{author}{Kosloff, R.}
\newblock \bibinfo{title}{A quantum mechanical open system as a model of a heat
  engine}.
\newblock \emph{\bibinfo{journal}{J. Chem. Phys.}}
  \textbf{\bibinfo{volume}{80}}, \bibinfo{pages}{1625--1631}
  (\bibinfo{year}{1984}).

\bibitem{Gemmer}
\bibinfo{author}{Gemmer, J.}, \bibinfo{author}{Michel, M.} \&
  \bibinfo{author}{Mahler, G.}
\newblock \emph{\bibinfo{title}{Quantum thermodynamics}}
  (\bibinfo{publisher}{Springer}, \bibinfo{year}{2009}).

\bibitem{ControlGelbwaser}
\bibinfo{author}{Gelbwaser-Klimovsky, D.}, \bibinfo{author}{Niedenzu, W.} \&
  \bibinfo{author}{Kurizki, G.}
\newblock \bibinfo{title}{Thermodynamics of quantum systems under dynamical
  control}.
\newblock \emph{\bibinfo{journal}{Adv. At. Mol. Opt. Phys.}}
  \textbf{\bibinfo{volume}{64}}, \bibinfo{pages}{329} (\bibinfo{year}{2015}).

\bibitem{Oppenheim02}
\bibinfo{author}{Oppenheim, J.}, \bibinfo{author}{Horodecki, M.},
  \bibinfo{author}{Horodecki, P.} \& \bibinfo{author}{Horodecki, R.}
\newblock \bibinfo{title}{Thermodynamical approach to quantifying quantum
  correlations}.
\newblock \emph{\bibinfo{journal}{Phys. Rev. Lett.}}
  \textbf{\bibinfo{volume}{89}}, \bibinfo{pages}{180402}
  (\bibinfo{year}{2002}).

\bibitem{AlickiFannes13}
\bibinfo{author}{Alicki, R.} \& \bibinfo{author}{Fannes, M.}
\newblock \bibinfo{title}{Entanglement boost for extractable work from
  ensembles of quantum batteries}.
\newblock \emph{\bibinfo{journal}{Phys. Rev. E}} \textbf{\bibinfo{volume}{87}},
  \bibinfo{pages}{042123} (\bibinfo{year}{2013}).

\bibitem{Acinentanglement}
\bibinfo{author}{Hovhannisyan, K.~V.}, \bibinfo{author}{Perarnau-Llobet, M.},
  \bibinfo{author}{Huber, M.} \& \bibinfo{author}{Ac\'in, A.}
\newblock \bibinfo{title}{Entanglement generation is not necessary for optimal
  work extraction}.
\newblock \emph{\bibinfo{journal}{Phys. Rev. Lett.}}
  \textbf{\bibinfo{volume}{111}}, \bibinfo{pages}{240401}
  (\bibinfo{year}{2013}).

\bibitem{Perarnau15}
\bibinfo{author}{Perarnau-Llobet, M.} \emph{et~al.}
\newblock \bibinfo{title}{Extractable work from correlations}.
\newblock \emph{\bibinfo{journal}{Phys. Rev. X}} \textbf{\bibinfo{volume}{5}},
  \bibinfo{pages}{041011} (\bibinfo{year}{2015}).

\bibitem{Campisi16}
\bibinfo{author}{Campisi, M.} \& \bibinfo{author}{Fazio, R.}
\newblock \bibinfo{title}{The power of a critical heat engine}.
\newblock \emph{\bibinfo{journal}{Nature Comm.}} \textbf{\bibinfo{volume}{7}},
  \bibinfo{pages}{11895} (\bibinfo{year}{2016}).

\bibitem{Campisi16b}
\bibinfo{author}{Campisi, M.} \& \bibinfo{author}{Fazio, R.}
\newblock \bibinfo{title}{Dissipation, correlation and lags in heat engines}
  (\bibinfo{year}{2016}).
\newblock \eprint{arXiv:1603.05029}.

\bibitem{Zheng15}
\bibinfo{author}{Zheng, Y.} \& \bibinfo{author}{Poletti, D.}
\newblock \bibinfo{title}{Quantum statistics and the performance of engine
  cycles}.
\newblock \emph{\bibinfo{journal}{Phys. Rev. E}} \textbf{\bibinfo{volume}{92}},
  \bibinfo{pages}{012110} (\bibinfo{year}{2015}).

\bibitem{Janzing00}
\bibinfo{author}{Janzing, D.}, \bibinfo{author}{Wocjan, P.},
  \bibinfo{author}{Zeier, R.}, \bibinfo{author}{Geiss, R.} \&
  \bibinfo{author}{Beth, T.}
\newblock \bibinfo{title}{{Thermodynamic cost of reliability and low
  temperatures: Tightening Landauer's principle and the second law}}.
\newblock \emph{\bibinfo{journal}{Int. J. Th. Phys.}}
  \textbf{\bibinfo{volume}{39}}, \bibinfo{pages}{2717} (\bibinfo{year}{2000}).

\bibitem{Note1}
\bibinfo{note}{If the subsystems do not interact -- or they do it in the
  weak-coupling regime -- the set of possible dynamics is trivially given by
  products of local unitaries. In this regime the limitations of local control
  are not present.}

\bibitem{Lloyd1995}
\bibinfo{author}{Lloyd, S.}
\newblock \bibinfo{title}{Almost any quantum logic gate is universal}.
\newblock \emph{\bibinfo{journal}{Phys. Rev. Lett.}}
  \textbf{\bibinfo{volume}{75}}, \bibinfo{pages}{346--349}
  (\bibinfo{year}{1995}).

\bibitem{Ramakrishna1995}
\bibinfo{author}{Ramakrishna, V.}, \bibinfo{author}{Salapaka, M.~V.},
  \bibinfo{author}{Dahleh, M.}, \bibinfo{author}{Rabitz, H.} \&
  \bibinfo{author}{Peirce, A.}
\newblock \bibinfo{title}{Controllability of molecular systems}.
\newblock \emph{\bibinfo{journal}{Phys. Rev. A}} \textbf{\bibinfo{volume}{51}},
  \bibinfo{pages}{960€--966} (\bibinfo{year}{1995}).

\bibitem{Albertini2002}
\bibinfo{author}{Albertini, F.} \& \bibinfo{author}{D'Alessandro, D.}
\newblock \bibinfo{title}{{The Lie algebra structure and controllability of
  spin systems}}.
\newblock \emph{\bibinfo{journal}{Lin. Alg. App.}}
  \textbf{\bibinfo{volume}{350}}, \bibinfo{pages}{213--€"235}
  (\bibinfo{year}{2002}).

\bibitem{Janzing2002}
\bibinfo{author}{Janzing, D.}, \bibinfo{author}{Armknecht, F.},
  \bibinfo{author}{Zeier, R.} \& \bibinfo{author}{Beth, T.}
\newblock \bibinfo{title}{Quantum control without access to the controlling
  interaction}.
\newblock \emph{\bibinfo{journal}{Phys. Rev. A}} \textbf{\bibinfo{volume}{65}}
  (\bibinfo{year}{2002}).

\bibitem{Lloyd2004}
\bibinfo{author}{Lloyd, S.}, \bibinfo{author}{Landahl, A.~J.} \&
  \bibinfo{author}{Slotine, J.-J.~E.}
\newblock \bibinfo{title}{Universal quantum interfaces}.
\newblock \emph{\bibinfo{journal}{Phys. Rev. A}} \textbf{\bibinfo{volume}{69}}
  (\bibinfo{year}{2004}).

\bibitem{Burgarth2009}
\bibinfo{author}{Burgarth, D.}, \bibinfo{author}{Bose, S.},
  \bibinfo{author}{Bruder, C.} \& \bibinfo{author}{Giovannetti, V.}
\newblock \bibinfo{title}{Local controllability of quantum networks}.
\newblock \emph{\bibinfo{journal}{Phys. Rev. A}} \textbf{\bibinfo{volume}{79}}
  (\bibinfo{year}{2009}).

\bibitem{Wilming2016}
\bibinfo{author}{Wilming, H.}, \bibinfo{author}{Gallego, R.} \&
  \bibinfo{author}{Eisert, J.}
\newblock \bibinfo{title}{Second law of thermodynamics under control
  restrictions}.
\newblock \emph{\bibinfo{journal}{Phys. Rev. E}} \textbf{\bibinfo{volume}{93}},
  \bibinfo{pages}{042126} (\bibinfo{year}{2016}).

\bibitem{Note2}
\bibinfo{note}{Indeed, this is not enough, as we will see in the ferromagnetic
  case}.

\end{thebibliography}

\newpage
\appendix

\section{General bound}
\subsection{Proof of Theorem \ref{thm:general-results}}
\label{sec:app:penalty}
Let us consider any protocol that makes a cycle with $n_h$ contacts of the form \eqref{eq:thermalcontact} with the bath at $\beta_h>0$, followed by $n_c$ contacts with the bath at $\beta_c>0$. The fact that the optimal protocol has to be of this form follows easily from the following considerations. 

Let us denote the initial Hamiltonian by $H^D$ and the Hamiltonian after the $n_h$-th contact with the bath at $\beta_h$ as $H^B$ (see Fig \ref{fig:carnotcyclequbit} a)). Let us also denote by $W^{D\rightarrow B}$ and $Q^{D\rightarrow B}$ the 
expected work and heat respectively obtained in this part of the protocol between $H^D$ and $H^B$. We recall from Theorem 1 in  Ref. \cite{Wilming2016} that the optimal value of $W^{D\rightarrow B}$ can be written as
\begin{equation}\label{eq:proof1}
W^{D\rightarrow B}\leq T_h D(\omega_{\beta_c}^D\|\omega_{\beta_h}^B) - T_h D(V\omega_{\beta_c}^DV^{\dagger} \| \omega_{\beta_h}^A ).
\end{equation}
There, it is also shown how this value of work can be achieved, namely, by performing an adiabatic operation of the form \eqref{eq:timeevolution} as $(\omega_D^C, H^D) \rightarrow (V \omega_D^C V^{\dagger}, H^A)$, followed by an isothermal path from $H^A$ to $H^B$. Due to energy conservation, this strategy which maximizes $W^{D \rightarrow B}$ minimizes also $Q^{D\rightarrow B}$. By using the correspondence $D(\rho||\omega_\beta(H))=\beta\tr((\rho-\omega_\beta(H))H)-(S(\rho)-S(\omega_\beta(H)))$ and the first law of thermodynamics, we hence obtain
\begin{equation}\label{eq:proof2}
Q^{D \rightarrow B} \geq T_h \left(\Delta S^{B,D}- D(V\omega_c^DV^\dagger||\omega_h^A)\right).
\end{equation}
Using an equivalent argument to the one leading to \eqref{eq:proof1}, one finds 
that the work $W^{B \rightarrow D}$ from $H^B$ to $H^D$ and where only the cold bath is employed is bounded by
\begin{equation}\label{eq:proof3}
W^{B\rightarrow C}\leq T_cD(\omega_{\beta_h}^B\|\omega_{\beta_c}^D) - T_c D(U\omega_{\beta_h}^BU^{\dagger} \| \omega_{\beta_c}^C ).
\end{equation}
It then amounts to simple algebra to compute that
\begin{eqnarray}
\eta &=& \frac{W^{D\rightarrow B}+W^{B\rightarrow C}}{Q^{D \rightarrow B}}\nonumber \\
\label{eq:finaleqproof} &\leq& 1 - \dfrac{T_c}{T_h} \left(\dfrac{\Delta S^{B,D} + D(U\omega_{h}^{B}U^\dagger||\omega_{c}^{C}) }{\Delta S^{B,D} -  D(V\omega_{c}^{D}V^\dagger||\omega_{h}^{A})}\right).
\end{eqnarray}
The correction terms  $D(U\omega_{h}^{B}U^\dagger||\omega_{c}^{C})$ and $D(V\omega_{c}^{D}V^\dagger||\omega_{h}^{A})$ emerge as result of the unavoidable dissipation when switching from the hot to the cold bath 
or vice versa. Hence, it is easy to see that any other protocol that would include longer sequences of switches between the baths (see for instance Fig.\ \ref{fig:carnotcyclequbit} b)) would contain more dissipation terms that would diminish the efficiency even further. Hence, \eqref{eq:finaleqproof} provides the final bound and proves the
validity of Theorem \ref{thm:general-results}.

\subsection{Inequality regarding relative entropy}
\label{sec:app:relent}
We will now turn to proving the inequality 
\begin{equation}
D(U\omega_{\beta_1}(H_1)U^\dagger || \omega_{\beta_2}(H_2)) \geq D^\downarrow(\omega_{\beta_1}(H_1)||\omega_{\beta_2}(H_2)). 
\end{equation}
To do that we use the correspondence with the free energy and write the left hand side as
\begin{align}
&\beta_2\tr\left(H_2U\omega_{\beta_1}(H_1)U^\dagger\right) - S(U\omega_{\beta_1}U^\dagger) + \log Z_{\beta_2}(H_2)\nonumber\\
&= \beta_2\tr\left(H_2U\omega_{\beta_1}(H_1)U^\dagger\right) - S(\omega_{\beta_1}) + \log Z_{\beta_2}(H_2).
\end{align}
The r.h.s. corresponds to the case where $U$ is chosen such that $U\omega_{\beta_1}(H_1)U^\dagger$ is diagonal in the basis of $H_2$ with larger eigenvalues corresponding to smaller energies. It is thus the corresponds to the choice of $U$ that minimizes the energy. Hence the l.h.s. is always as least as big as the r.h.s.

\section{Ising model}
\subsection{Ground-state degeneracy}
\label{sec:ground state-deg}
In this section, we discuss the ground state degeneracy of the nearest-neighbour Ising model with finite interactions and 
magnetic field. The goal is to explain the finite entropy density, and hence finite work-density, that can be reached
for any interaction strength in case of anti-ferromagnetic couplings. We will therefore restrict to this scenario. First, for any fixed temperature and zero magnetic field the thermal state converges to
the thermal ground state as the interaction strength $J$ is made large 
in absolute value. Similarly, if the magnetic field has a strength $h = kJ$, the thermal state approximates the ground state with unit interaction
strength and magnetic field of strength $k$. For large interaction strengths, 
we are thus interested in the entropy density at zero temperature in the case $J=1$ and $h=k$. We therefore have to count the ground state degeneracy in such a situation. Clearly, for $k\rightarrow \infty$ the ground state degeneracy is finite, i.e., independent of the system size. The same holds for $k=0$ and an even number of spins while for $k=0$ and an odd number of spins, the ground state degeneracy scales linearly with the system size. In all cases, the entropy density vanishes in the thermodynamic limit. We will now show that there are finite values of $k$ such that the ground state degeneracy is exponentially large in the system size, so that the entropy density remains finite in the thermodynamic limit. 

To see this, suppose that we set $k$ to be given by the number of nearest neighbours of a single site in the lattice. For a square lattice we thus have $k=2d$, with $d$ the spatial dimension. It is easy to convince oneself that one of the 
ground states is given by $\ket{\!\uparrow}^{\otimes N}$, where $\ket{\!\uparrow}$ denotes the spin-up state vector in the direction of the magnetic field.  It has energy 
\begin{equation}
dN-2dN = -dN.
\end{equation} 

Now suppose that we flip one of the spins. The increase of energy due to the magnetic field is given by $2k=4d$, and the interaction energy of each of the $k$ neighbours is reduced by $2J=2$. The net change of energy is 
found to be $2k-2k=0$ and we therefore have produced a new ground state. If we would flip a neighbour of the flipped spin, the energy would increase. However if we flip a next-nearest neighbour of the flipped spin, we obtain a further ground state. Iterating in the same way we get new ground states until we have flipped half of the spins in the lattice. However, we can always decide to leave out one of the nearest neighbours. In other words, we can decide for each of the $N/2$ next-nearest neighbours, whether we want to flip it, providing us with the lower bound of the ground state entropy
\begin{align}
S_{\mathrm{G}} \geq \log 2^{N/2} = N/2 \log 2.
\end{align}
Clearly, there are many more states with the same energy, for example the one corresponding to
$\ket{\uparrow ,\downarrow ,\uparrow ,\uparrow, \downarrow ,\uparrow,\cdots}$, which does not fit into the scheme described above. Nevertheless, our argument is sufficient to show that for $h=2JZ$, with $Z$ the coordination number of the lattice, we get a finite entropy-density in the zero-temperature state in the thermodynamic limit (and hence also at any positive finite-temperature). 

\subsection{Achievability of Carnot efficiency in the anti-ferromagnet}\label{sec:anti-ferromagnetic}
In this section, we show that Carnot efficiency is achievable at finite work per particle in the thermodynamic limit as $J\rightarrow -\infty$, that is, in the extremely anti-ferromagnetic case. Recall that the work-density in the thermodynamic limit is given by
\begin{align}
w(J) &= \lim_{n\rightarrow \infty} \frac{1}{n}(T_h-T_c)\Delta S^{B,D}  \\
&\quad - \lim_{n\rightarrow\infty} \frac{1}{n}\left( T_hD(\omega_{\beta_c}^{H_D}||\omega_{\beta_h}^{H_A}) + T_c D(\omega_{\beta_h}^{H_B}||\omega_{\beta_c}^{H_C})\right),\nonumber
\end{align}
with 
\begin{align}
\Delta S^{B,D} = S(\omega_{\beta_h}^{H_B})-S(\omega_{\beta_c}^{H_D}).
\end{align}
Here, the Hamiltonians $H_{A,B,C,D}$ correspond to different magnetic fields $h_{A,B,C,D}$ at the different stages of the protocol. 
The efficiency is given by
\begin{align}
\eta(J) = 1- \frac{T_c}{T_h}\frac{\Delta S^{B,D}  + D(\omega_{\beta_h}^{H_B}||\omega_{\beta_c}^{H_C}) }{\Delta S^{B,D} -  D(\omega_{\beta_c}^{H_D}||\omega_{\beta_h}^{H_A})}.
\end{align}
In the formula for the efficiency we have, for notational reasons, omitted the thermodynamic limit $n\rightarrow \infty$. 
Carnot efficiency is reached only if the two correction-terms involving relative entropies vanish \footnote{Indeed, this is not enough, as we will see in the ferromagnetic case}. The work-density depends on $J$ and the external fields through the Hamiltonians $H_{A,B,C,D}$. In the limit $J\rightarrow -\infty$, the thermal states $\omega_{\beta}^H$ converge to ground state projectors. We will now choose $h_C=h_B= 2J$ and $h_D=h_A\gg J\rightarrow \infty$ (compare 
this with the results of Section \ref{sec:ground state-deg}). It is then clear that the relative entropy density including Hamiltonians $H_A$ and $H_D$ vanishes as it compares the state with all spins up with itself. 

Similarly, in the limit $J\rightarrow - \infty$, the relative entropy involving $H_B=H_C$ vanishes, 
because the two states at different temperatures both converge to the ground state of a Hamiltonian with a finite $J<0$. 
Note, however, that in this case the two states converge to the ground state of a model in which $h=2J$ and hence have finite entropy-density. 

Recall that the anti-ferromagnetic ground state at infinite external field has a unique ground state, whereas for $h=2J$, the ground state space is exponentially degenerate in the system size and therefore has a finite entropy density (see section~\ref{sec:ground state-deg}). Combining with the previous considerations, we have
\begin{align}
\lim_{J\rightarrow -\infty} w(J) &=  \lim_{J\rightarrow -\infty}\lim_{n\rightarrow \infty} \frac{1}{n}(T_h-T_c)\left(\Delta S^{B,D}\right)\nonumber \\
&= \lim_{J\rightarrow -\infty}\lim_{n\rightarrow \infty} \frac{1}{n}(T_h-T_c)S(\omega_{\beta_h}^{H_B})\nonumber\\
&\geq (T_h-T_c) \frac{1}{2}\log(2),  
\end{align}
which is consistent with our numerics and the fact that we are in fact neglecting a number of ground states that is exponential in the system size in our estimate of the ground state degeneracy. 

Combining the fact that the entropy density $\Delta S^{B,D}/n$ remains finite with the observation 
that the corrections vanish, we obtain
\begin{align}
\lim_J \eta(J) = 1- \lim_J \frac{T_c}{T_h} \frac{\Delta S^{B,D}}{\Delta S^{B,D}} = 1-\frac{T_c}{T_h}.
\end{align}

\subsection{Ferromagnetic case: vanishing efficiency}\label{sec:ferromagnetic}
In the ferromagnetic case, the optimal protocol in terms of work-density is to choose $h_A = h_D \rightarrow \infty$ and $h_B=h_C=0$ in the sense that $h_A,h_D \gg J$ for any choice of $J$. This is proven in Section~\ref{sec:nonsmooth}.
Although this is the optimal protocol, it achieves zero work-density in the strong coupling limit $J\rightarrow \infty$. This is due to the fact that the ground state degeneracy vanishes in the ferromagnetic case for any choice of magnetic field and for any finite $J>0$. 
We will now argue that this optimal protocol not only has zero work-yield, but also vanishing efficiency as $J\rightarrow \infty$ despite the fact that both penalty terms become zero as $J\rightarrow \infty$. 

For this discussion, it is useful to consider not the actual optimal protocol, but allow the fields $h_B,h_C$ to depart 
from $0$ by some amount $\epsilon>0$, representing the precision with which we can control the magnetic field in the experiment.  In the following it is also useful to keep in mind that the temperature-difference $\Delta T$ is fixed, and hence $\Delta T/ J \rightarrow 0$, so that $\Delta T$ can be seen as arbitrarily small as $J$ goes to infinity.
Now first notice that since $h_A,h_D\rightarrow \infty$, one of the penalty terms in the efficiency can effectively be set to zero just as in the anti-ferromagnetic case. Furthermore $S(\omega^{H_D}_{\beta_c})$ vanishes, so that $\Delta S^{B,D}$ can be replaced by $S(\omega^{H_B}_{\beta_h})=: S^B_h$. The efficiency then takes the form
\begin{align}
\lim_{J}\eta(J) &= \lim_J \left(1- \frac{T_c}{T_h}\frac{S^B_h + D(\omega_{\beta_h}^{H_B}||\omega_{\beta_c}^{H_B}) }{S^B_h}\right).
\end{align}
We thus have to show that the second term converges to unity. 
To do that, first write the relative entropy as a difference of free energies and cancel entropic terms,
\begin{align}
\frac{T_c}{T_h}\frac{S^B_h + D(\omega_{\beta_h}^{H_B}||\omega_{\beta_c}^{H_B}) }{S^B_h}  
&= 
\frac{T_c}{T_h}\frac{S^B_h + \frac{1}{T_c}(E^B(T_h) - T_c S_h - F^B_c)}{S^B_h} \nonumber \\
&= 
\frac{1}{T_h}\frac{E^B(T_h) - F^B_c}{S^B_h} \nonumber \\
&= 
1 + \frac{1}{T_h}\frac{F^B_h - F^B_c}{S^B_h} .
\end{align}
Here, we have introduced the internal energy with respect to the Hamiltonian $H_B$ and at temperature $T$ as $E^B(T)$ and the thermal free energies at the temperatures $T_h$ and $T_c$ as $F^B_{h,c}$, respectively. They fulfill
\begin{align}
F^B_{h,c} = J f(T_{c,h}/J), \quad f(T) := - T \log Z_T(1,\epsilon/J),
\end{align}
where $Z_T(J,h)$ is the partition function of the model at temperature $T$, interaction 
strength $J$ and magnetic field strength $h$.
We can then expand $F^B_{h}$ in the small parameter $(T_h-T_c)/J=:\Delta T/J$ around $T_c/J$, to obtain
\begin{align}
F^B_h - F^B_c &= J\left(\left.\frac{\partial f(x)}{\partial x}\right|_{x=T_c/J}\frac{T_h-T_c}{J} + O(1/J^2)\right) \nonumber \\
&= - S^B_c (T_h-T_c) + O(1/J). 
\end{align}
Using this result,  then yields for the efficiency
\begin{align}
\lim_j \eta(J) \nonumber &= \lim_{J\rightarrow \infty }
\left( \frac{S^B_c}{S^B_h}\frac{\Delta T}{T_h} + O(1/J)\right) \nonumber \\
&= \frac{\Delta T}{T_h}\frac{\log(1+\e^{ - \beta_c \epsilon N})}{\log(1+\e^{- \beta_h \epsilon N})}\leq  \frac{\Delta T}{T_h},
\end{align}
where the last line is proven in the following section. We thus see that in any finite system it is formally possible to achieve Carnot-efficiency in the limit $J\rightarrow \infty$ if we can get $\epsilon$ exactly to zero. However, as the system size increases, to achieve a given efficiency, the precision has to scale like $1/N$. If we have a fixed precision, the efficiency goes to zero exponentially. We thus conclude that it is physically infeasible to achieve finite efficiency in the thermodynamic limit.
In a finite system, in contrast, the efficiency can be made as close to Carnot efficiency as $J\rightarrow \infty$ by increasing the precision. This is illustrated in Fig.~\ref{fig:precision} for a system of six spins.

\begin{figure}
\includegraphics[width=0.68\columnwidth]{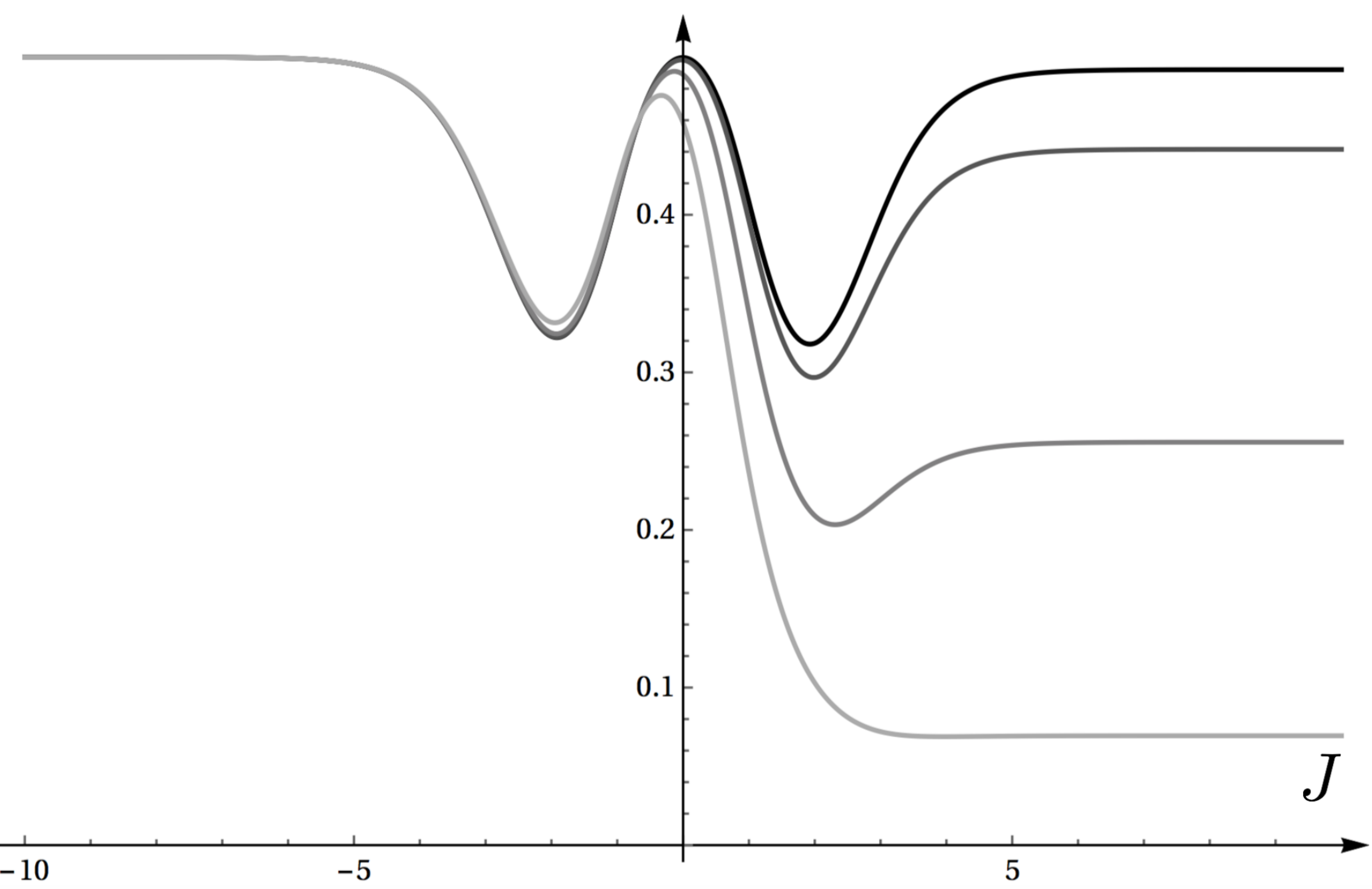}
\caption{Efficiency at maximum work for a system of six spins as a function of $J$ and different imprecisions on the external fields $\epsilon=0.05,0.1,0.25,0.5$ (black to light grey). For larger system sizes, the local minimum on the ferromagnetic side $J>0$ moves to larger values of $J$ and the value at $J=\infty$ decreases exponentially with the system size for any fixed precision. }
\label{fig:precision}
\end{figure}

\subsection{Ratios of entropies}
We will now discuss the ratios of entropies appearing in the ferromagnetic strong-coupling setting, first in the
situation in which we can assume that we may set the magnetic field exactly to zero. We will then turn to 
elaborating on the case of a small but finite external field. We will keep the discussion as general as
possible, that is, we will use hardly any specific properties of the Ising model apart from its ground 
state properties and the Hamiltonian gap. 
\begin{lemma}[Entropy ratios]
Let $H\geq 0$ be a Hamiltonian on a finite-dimensional system with at least two different energy-levels. Consider two inverse temperatures $\beta_c > \beta_h\geq 0$.
Then
\begin{align}
\lim_{J\rightarrow +\infty}\frac{S(\omega_{c}(J))}{S(\omega_{h}(J))} = \begin{cases}
0,\quad g(E_0) = 1,\\
1,\quad g(E_0) \geq 2.
\end{cases},
\end{align}
where $g(E_0)$ is the ground state degeneracy.
\end{lemma}
\begin{proof}
Here, we have used the Gibbs states
\begin{align}
\omega_{c/h}(J) = \frac{\e^{-\beta_{c/h} J H}}{Z_{c/h}(J)}.
\end{align}
Without loss of generality assume $H=\sum_{i=0}^{n-1} E_i P_i \geq 0$ with eigen-energies $E_i$, eigen-projectors $P_i$, ground state energy $E_0=0$ and degeneracies $g(E_i) = \tr(P_i)$. Then $Z_{c/h}(J) > g(E_0) \geq 1$ and $\lim_{J\rightarrow \infty} Z_{c/h}(J)=g(E_0)$. A simple calculation gives
\begin{align}
\frac{S(\omega_{\beta_c}(JH))}{S(\omega_{\beta_h}(JH))} &= \frac{\beta_c J \tr(\omega_{c}H) + \log Z_c(J)}{\beta_h J \tr(\omega_h H) + \log Z_h} \\
&\leq \frac{\beta_c \tr(\omega_{c}H)}{\beta_h  \tr(\omega_h H)} + \frac{\log Z_c(J)}{\log Z_h (J)},\nonumber§
\end{align}
where we have used that both terms in the denominator are positive. We will now show that both terms vanish as $J\rightarrow \infty$. Starting with the first term, we get
\begin{align}
\frac{\beta_c \tr(\omega_{c}H)}{\beta_h J \tr(\omega_h H)} &\leq \frac{Z_h(J) \beta_c \tr(\omega_{c}H)}{J \e^{-\beta_h J E_1} g(E_1) \beta_h E_1} \\
&= \frac{ Z_h(J)  \sum_i \e^{-\beta_c J E_i} \beta_c E_i g(E_i)} {Z_c(J) \e^{-\beta_h J E_1} g(E_1) \beta_h E_1} \nonumber 
\\
&\leq  \frac{ Z_h(J) \e^{-\beta_c J E_1} \sum_i  \beta_c E_i g(E_i)} { Z_c(J)\e^{-\beta_h J E_1} g(E_1) \beta_h E_1}  \nonumber\\
&\leq \e^{-JE_1 (\beta_c - \beta_h)}Z_h(J) K , \nonumber
\end{align}
where $K>0$ is some  constant, independent of $J$. Thus, as $J\rightarrow \infty$, the term goes to zero exponentially since $Z_h(J)\rightarrow g(E_0)$.
Let us now consider the second term. 
In the case $g(E_0) \geq 2$, we have $\lim_{J}\log Z_{c/h}(J) = \log g(E_0) \geq \log 2$ and the ratio converges to $1$. 
Let us, therefore, assume that $g(E_0)=1$. We first employ the fact that the logarithm is monotone increasing to truncate the partition sum in the denominator. This yields
\begin{align}
\frac{\log Z_c(J)}{\log Z_h (J)} &\leq \frac{\log Z_c(J)}{\log g(E_0)+\log\left(1+ \frac{g(E_1)}{g(E_0)} \e^{-\beta_h J E_1} \right)} \\
&= \frac{\log Z_c(J)}{\log\left(1+ \overline{g}_1 \e^{-\beta_h J E_1} \right)},\nonumber
\end{align}
where we have written $\overline{g}_i := g(E_i)/g(E_0)$. We now use that $E_i \geq E_0$ for $i\geq 1$ and the monotonicity of the logarithm again to upper bound the numerator as
\begin{align}
\frac{\log Z_c(J)}{\log Z_h (J)} &\leq \frac{\log g(E_0) + \log\left(1 + \e^{-\beta_c J E_1} \sum_i \overline{g}_i\right)}{\log\left(1+ \overline{g}_1 \e^{-\beta_h J E_1} \right)}\nonumber \\
&= \frac{\log\left(1 + \e^{-\beta_c J E_1} \sum_i \overline{g}_i\right)}{\log\left(1+ \e^{-\beta_h J E_1} \overline{g}_1\right)}\\
&\leq \frac{\e^{-\beta_c J E_1} \sum_i \overline{g}_i}{\log\left(1+ \e^{-\beta_h J E_1} \overline{g}_1\right)},\nonumber
\end{align}
where we have used $\log x \leq x-1$. For large $x$ we have
\begin{align}
\log(1+ \e^{-a x}C) \simeq \e^{- ax}C.
\end{align}
Hence, we finally obtain 
\begin{align}
\lim_J\frac{\log Z_c(J)}{\log Z_h (J)} \leq \lim_J \frac{\sum_i \overline{g}_i}{\overline{g}_1}\e^{-(\beta_c-\beta_h) J E_1} = 0.
\end{align}
\end{proof}

Now consider the Hamiltonian 
\begin{equation}
	H(J) = H + \frac{B}{J} V, 
\end{equation}
where $H$ is a local Hamiltonian on $N$ sites, has gap of order unity and a two-fold degenerate ground state. 
Furthermore, suppose that $V$ is also a local Hamiltonian which merely splits the ground state degeneracy of $H$ by an amount $N B/J$ for large enough $J$, but does not change the order of the gap. Then by a similar reasoning as in the previous lemma we obtain
\begin{align}
\overline{\eta}(B,N) :=\lim_{J\rightarrow \infty} \frac{S(\omega_{c}(J))}{S(\omega_{h}(J))} = \frac{\log(1+\e^{- \beta_c B N})}{\log(1+\e^{- \beta_h B N})} < 1,
\end{align}
where now
\begin{align}
\omega_{c/h}(J) = \frac{\e^{-\beta_{c/h} J H(J)}}{Z_{c/h}(J)}.
\end{align}

\subsection{Optimal protocols in the thermodynamic limit}\label{sec:nonsmooth}

In this section, we continue discussing the example of the Ising model and show that both the optimal work-density and the efficiency at optimal work-density are not smooth at $J=J^*$. We will work directly in the thermodynamic limit, where the free energy density takes the well-known form
\begin{align}
f(\beta,J,h) = - \frac{1}{\beta}\log&\left(\e^{\beta J}\cosh(\beta h)\right.\\ &\quad\left.+
({\e^{2\beta J}\sinh(\beta h)^2+\e^{-2\beta J}})^{1/2}\right).\nonumber
\end{align}
It is clear from the discussion of the main text  that in order 
to optimize the work-density, we have to maximize the entropy-density as a function of the magnetic field for a given $J$. Here, we are interested in the anti-ferromagnetic regime, i.e., $J<0$. 
The entropy-density can be calculated from the above expression explicitly by the usual formula
\begin{align}
s(\beta,J,h) = -\frac{\partial}{\partial T}f(1/T,J,h),
\end{align}
resulting, however, in a fairly complicated expression. To find an extremum of the entropy-density as a function of $h$, we take the corresponding derivative. The result is
\begin{widetext}
\begin{align}
\frac{\partial s(\beta,J,h)}{\partial h}= - \beta^2\frac{\e^{\beta J}\left(h \cosh(\beta h)+2 J \sinh(\beta h)\right)}{\left({\e^{-2\beta J}+\e^{2\beta J}\sinh(\beta h)^2}\right)^{1/2}\left(1+\e^{4\beta J}\sinh(\beta h)^2\right)}. 
\end{align}
\end{widetext}
For this expression to vanish, we either need $h\rightarrow \infty$, so that the denumerator diverges, or that the numerator vanishes. The former case corresponds to vanishing entropy-density, as it corresponds to a magnetic field that is so strong that it projects all spins in the same direction. We thus consider the second case in which we have to find functions 
with the property that \begin{align}
h(J) \cosh(\beta h(J))+2J \sinh(\beta h(J)) = 0.
\end{align}

\begin{figure}
\includegraphics[width=.8\linewidth]{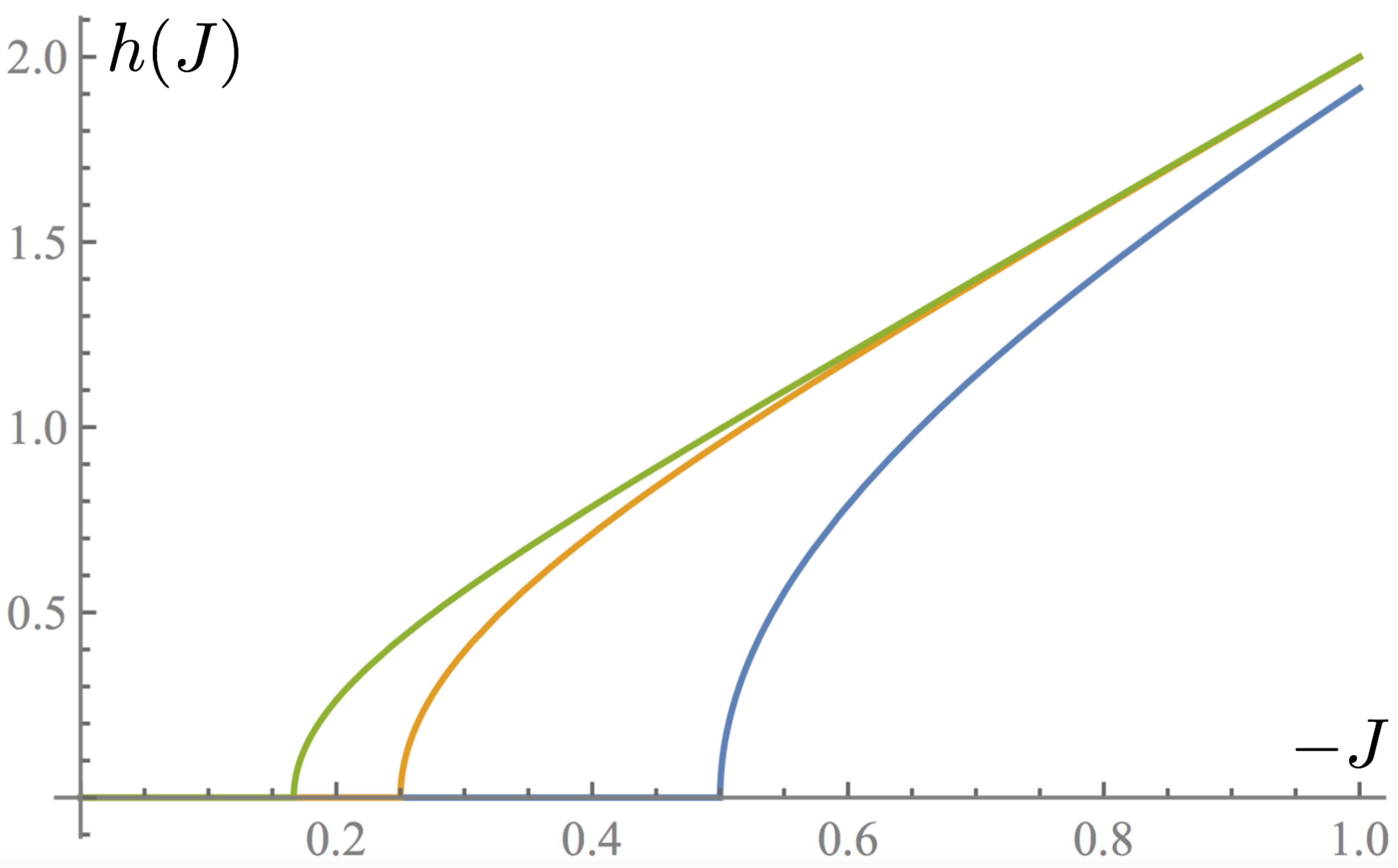}
\caption{The optimal magnetic field as a function of $J$ for inverse temperatures $\beta=1$ (blue), $\beta=2$ (orange) and $\beta=3$ (green). It is clearly visible that at the critical point $1/(2\beta)$ the function is not analytic, similarly to a second order phase transition.}
\label{fig:optimalmagneticfield}
\end{figure}

Clearly, one solution to this equation is given by $h_1(J)=0$. However, there can be more solutions.
Remembering that $J<0$, we can simplify this expression to
\begin{align}
h(J) = 2|J|\tanh(\beta h(J)).
\end{align}
The existence of a second solution $h_2$ now follows from the fact that $h\mapsto \tanh(\beta h)$ is concave for $h>0$ and convex for $h<0$, with derivative at the origin given by $\beta>0$. Thus as long as $2|J|\beta > 1$, or, 
in other words, 
\begin{equation}
|J| > \frac{1}{2}k_B T, 
\end{equation}
there exists a second solution to the equation. It is also clear that this second solution only exists for $J<0$. We have plotted the optimal magnetic field in Fig.\ \ref{fig:optimalmagneticfield}. It is clearly not continuously differentiable. 

Finally, we note that the solution $h_2$ always provides a larger entropy than the trivial solution $h_1(J)=0$. From the discussion of the ground state entropy in the anti-ferromagnetic case as a function of $h$, 
we can guess that for large $\beta$, the optimal magnetic field is given by $2|J|$. Indeed, we have
\begin{align}
2|J|\left(1-\tanh(\beta 2|J|)\right) \rightarrow 0 
\end{align}
as $ |J|\rightarrow \infty$,
showing that for very strong anti-ferromagnetic interactions $h(J)=2|J|$ is arbitrary close to the optimal value $h_2$. \clearpage

\end{document}